\documentclass[12pt]{amsart}
\usepackage{graphicx,amscd, color,amsmath,amsfonts,amssymb,geometry,amssymb}
\usepackage[initials]{amsrefs}

\newtheorem{theorem}{Theorem}[section]
\newtheorem{algorithm}[theorem]{Algorithm}

\newtheorem{corollary}[theorem]{Corollary}

\newtheorem{examples}[theorem]{Examples}
\newtheorem{remark}[theorem]{Remark}

\newtheorem{lemma}[theorem]{Lemma}

\newtheorem{definition}[theorem]{Definition}
\numberwithin{equation}{section}

\begin{document}

\title{Sinkhorn-Knopp Theorem for rectangular positive maps}

\author[Cariello ]{D. Cariello}

\address{Faculdade de Matem\'atica, \newline\indent Universidade Federal de Uberl\^{a}ndia, \newline\indent 38.400-902 Ð Uberl\^{a}ndia, Brazil.}\email{dcariello@ufu.br}


\subjclass[2010]{}

\maketitle

\begin{abstract} 
In this work, we adapt Sinkhorn-Knopp theorem for rectangular positive maps  $(T:M_k\rightarrow M_m)$. We extend their concepts of support and total support  to these maps. We show that a positive map $T:M_k\rightarrow M_m$ is equivalent to a doubly stochastic map if and only if $T:M_k\rightarrow M_m$ is equivalent to a positive map with total support. 
 Moreover, if $k$ and $m$ are coprime then $T:M_k\rightarrow M_m$ is equivalent to a doubly stochastic map if and only if $T:M_k\rightarrow M_m$ has support. 

This result provides a necessary and sufficient condition for the filter normal form, which is commonly used in Quantum Information Theory in order to simplify the task of detecting entanglement. 
Let $A=\sum_{i=1}^nA_i\otimes B_i\in M_k\otimes M_m$ be a state and $G_A: M_k\rightarrow M_m$ be the positive map $G_A(X)=\sum_{i=1}^nB_itr(A_iX)$. We show that $A$  can be put in the filter normal form if and only if $G_A: M_k\rightarrow M_m$ is equivalent to a positive map with total support.
We prove that any state $A\in M_k\otimes M_m\simeq M_{km}$ such that $\dim(\ker(A))<k-1$, if $k=m$, and $\dim(\ker(A))<\min\{k,m\}$, if $k\neq m$, can be put in the filter normal form. 

Recently, a connection between the capacity of a rectangular positive map $T:M_k\rightarrow M_m$ and the capacity of a certain square positive map $\widetilde{T}:M_{mk}\rightarrow M_{mk}$ was noticed. Here, we obtain a deeper connection between these maps. As a corollary of our main results, we prove that $T:M_k\rightarrow M_m$ is equivalent to a doubly stochastic map if and only if  $\widetilde{T}:M_{mk}\rightarrow M_{mk}$ is equivalent to a doubly stochastic map.



\end{abstract}

\section*{Introduction}

The Sinkhorn-Knopp theorem states that there are positive diagonal matrices $D_1,D_2$ such that $D_1MD_2$ is doubly stochastic if and only if the square matrix $M$ with non-negative entries has total support. In \cite{Sinkhorn}, the authors provide an iterative algorithm in order to obtain the doubly stochastic matrix from the original matrix. The convergence of this algorithm  was proved, whenever the original matrix has support.

There are generalisations of Sinkhorn-Knopp theorem \cite{Bapat, Brualdi, Git, Georgiou, gurvits2003, gurvits2004, filternormalform, leinaas, Sinkhorn2}.
 One of them is particularly important for Quantum Information Theory, it is the so-called filter normal form (see \cite[section IV.D]{Git} and \cite{filternormalform,leinaas}). The aim of this normal form is to put a mixed state in a certain format in order to simplify the task of detecting entanglement. Let $M_k$ and $P_k$ denote the set of complex matrices of order $k$ and the set of positive semidefinite Hermitian matrices of order $k$, respectively. In \cite{leinaas}, the authors used the filter normal form to prove the separability of every positive definite PPT state in $M_2\otimes M_2$. By a continuity argument, it implies that every PPT state in $M_2\otimes M_2$ is separable.  In \cite{Git}, it was used to unify several separability conditions. Furthermore, there are some inequalities that grant separability that are sharp or can be sharpen using the filter normal form \cite[Remark 64]{cariello1} .

  Originally published in \cite{filternormalform}, the filter normal form was only obtained  for positive definite mixed states. Here, we obtain a necessary and sufficient condition for an arbitrary mixed state in $M_k\otimes M_m$ to be put in this normal form.
 
Let us identify  the tensor product space $M_{k}\otimes M_{m}$ with $M_{km}$, via Kronecker product. Let a linear map $T:M_k\rightarrow M_m$ be  positive, if it sends positive semidefinite Hermitian matrices to positive semidefinite Hermitian matrices $($i.e., $T(P_k)\subset P_m)$.  A positive map $T:M_k\rightarrow M_m$ is called doubly stochastic if $T(\frac{Id}{\sqrt{k}})=\frac{Id}{\sqrt{m}}$ and $T^*(\frac{Id}{\sqrt{m}})=\frac{Id}{\sqrt{k}}$, where $T^*$ is the adjoint of $T$ with respect to the trace inner product. Let us say that  $T:M_k\rightarrow M_m$ and $T_1:M_k\rightarrow M_m$ are equivalent, if there are invertible matrices $X'\in M_k$, $Y'\in M_m$ such that $T_1(X)=Y'T(X'XX'^*)Y'^*$. Let  $tr(C)$ denote the trace of $C\in M_k$. 
 
Let $A\in M_k\otimes M_m\simeq M_{km}$ be a positive semidefinite Hermitian matrix. We shall say that $A$ can be put in the filter normal form if there are invertible matrices $X'\in M_k$, $Y'\in M_m$ such that $(X'\otimes Y')A(X'\otimes Y')^*=\sum_{i=1}^n C_i\otimes D_i$, $C_1=\frac{Id}{\sqrt{k}}$, $D_1=\frac{Id}{\sqrt{m}}$ and $tr(C_iC_j)=tr(D_iD_j)=0$, for every $i\neq j$.  

Now, the key observation is the following: $A=\sum_{i=1}^n A_i\otimes B_i$ can be put in the filter normal form if and only if the positive map $G_A: M_k\rightarrow M_m$, defined by $G_A(X)=\sum_{i=1}^nB_itr(A_iX)$, is equivalent to a doubly stochastic map (See  \ref{thmequivfilter} and \ref{maintheorem} for details). 

The problem of determining whether a  positive map $T:M_k\rightarrow M_k$  is equivalent to a doubly stochastic one was firstly addressed in \cite{gurvits2003, gurvits2004}. The author defined the notion of capacity and proved that a positive map is equivalent to a doubly stochastic one if and only if its capacity is positive and achievable \cite{gurvits2004}.  

It was also shown that a  positive map $T:M_k\rightarrow M_k$ is equivalent to a doubly stochastic map if and only if there are orthogonal projections $V_i,W_i\in M_k$, $1\leq i\leq s$, such that
 \begin{enumerate}
 \item $\mathbb{C}^k=\bigoplus_{i=1}^s\Im(V_i)=\bigoplus_{i=1}^s\Im(W_i)$ and, for every $i$,
 \item $T(V_iM_kV_i)\subset W_iM_kW_i$, where $V_iM_kV_i=\{V_iXV_i, X\in M_k\}$,
 \item $\text{rank}(W_i)=\text{rank}(V_i)$, 
 \item $\text{rank}(X)<\text{rank}(T(X))$,\\ $\forall  X\in V_iM_kV_i\cap P_k$ such that $0<rank(X)<\text{rank}(V_i)$.
 \end{enumerate}

Here, we obtain the following generalisation.

A positive map $T: M_k\rightarrow M_m$ is equivalent to a doubly stochastic map  if and only if 
there are orthogonal projections $V_i\in M_k$, $W_i\in M_m$, $1\leq i\leq s$, such that
 \begin{enumerate}
 \item $\mathbb{C}^k=\bigoplus_{i=1}^s\Im(V_i)$, $\mathbb{C}^m=\bigoplus_{i=1}^s\Im(W_i)$ and, for every $i$,
 \item $T(V_iM_kV_i)\subset W_iM_mW_i$, 
 \item $\frac{\text{rank}(W_i)}{\text{rank}(V_i)}=\frac{m}{k}$,
 \item $\text{rank}(X)\text{rank}(W_i)<\text{rank}(T(X))\text{rank}(V_i)$,\\ $\forall  X\in V_iM_kV_i\cap P_k$ such that $0<rank(X)<\text{rank}(V_i)$.
 \end{enumerate}
  
As a corollary, we prove that
a positive map $T: M_k\rightarrow M_m$ is equivalent to a doubly stochastic map  if and only if the positive map $\widetilde{T}: M_m\otimes M_k\rightarrow M_m\otimes M_k$  defined by  $$\widetilde{T}(\sum_{i,j=1}^m e_ie_j^t\otimes B_{ij})=T(\sum_{i=1}^m B_{ii})\otimes Id_{k\times k},$$
is also equivalent (\ref{corollarysnonquare}). Recall that $M_m\otimes M_k\simeq M_{mk}$. 

The connection between the capacity  of $T: M_k\rightarrow M_m$ and $\widetilde{T}: M_m\otimes M_k\rightarrow M_m\otimes M_k$ was noticed in \cite{garg2}. This corollary provides a  deeper connection between these two maps.

In order to obtain the aforementioned generalisation, we adapt Sinkhorn-Knopp original argument.   We extend the concepts of support and total support to rectangular matrices and to positive maps.  We also use their key lemma (See \ref{lemmasinkhorn}) and for each item of our lemma \ref{propertiesalgorithm}, there is a corresponding item in their classical proof.

It turns out that support still is a necessary condition for the equivalence of  a positive map $T:M_k\rightarrow M_m$ with a doubly stochastic one, but total support is only sufficient (differently from the original Sinkhorn-Knopp theorem. See \ref{remark3}). 
The exact condition is the following:  A positive map $T:M_k\rightarrow M_m$ is equivalent to a doubly stochastic map if and only if $T:M_k\rightarrow M_m$ is equivalent to a positive map with total support. Thus, a positive semidefinite Hermitian matrix $A\in M_k\otimes M_m$ can be put in the filter normal form if and only if the positive map $G_A: M_k\rightarrow M_m$ is equivalent to a positive map with total support. 

In general, this condition cannot be easily checked. However, if $k$ and $m$ are coprime then support and  total support are equivalent notions  (See \ref{examples1}). Thus, if $k$ and $m$ are coprime then $A\in M_k\otimes M_m$ can be put in the filter normal form if and only if  $G_A: M_k\rightarrow M_m$ has support (See \ref{corollary0}). We can determine whether an arbitrary positive map has support if a certain sequence of matrices converges to the identity  (See \ref{corollaryequivalencesupport}). There is also a polynomial-time algorithm in \cite{garg2} that can be adapted to check whether $G_A:M_k\rightarrow M_m$ has support (See \ref{remarkpolynomialtime}).

 Moreover, we also provide some easy sufficient conditions to guarantee that a state $A\in M_k\otimes M_m$ can be put in the filter normal form: If $k=m$ and $\dim(\ker(A))<k-1$, or if $k\neq m$ and $\dim(\ker(A))<\min\{k,m\}$, or  if $G_A(Id)$ and $G_A^*(Id)$ are invertible matrices and $\dim(\ker(A))<\frac{\max\{k,m\}}{\min\{k,m\}}$ (See \ref{theoremlowkernel} and \ref{theoremlowkernel2}).

This paper is organised as follows. In Section 1, we extend the concepts of support and total support to rectangular matrices (definition \ref{definitionsupportnonsquare}). In Section 2, we define positive maps with support and total support (definition \ref{definitionsupportformaps}). We describe an adaptation of Sinkhorn and Knopp algorithm  for positive maps $T:M_k\rightarrow M_m$ (algorithm \ref{algorithm}). In Section 3, we obtain a necessary and sufficient condition for the equivalence of a positive map with a  doubly stochastic one (theorems \ref{maintheorem1} and \ref{maintheorem}). In Section 4, we address the filter normal form problem for a positive semidefinite Hermitian matrix $A\in M_k\otimes M_m$. As a consequence of our main theorems, we prove that the existence of this form is equivalent to the existence of invertible matrices $X'\in M_k$, $Y'\in M_m$ such that $Y'G_A(X'(\cdot)X'^*)Y'^*$ has total support (theorem \ref{thmequivfilter}). We show that this last condition can be granted, if $k$ and $m$ are coprime and $G_A$ has support (corollary \ref{corollaryequivalencesupport}), or if $k=m$ and $\dim(\ker(A))<k-1$, or if $k\neq m$ and $\dim(\ker(A))<\min\{k,m\}$ or  if $G_A(Id)$ and $G_A^*(Id)$ are invertible matrices and $\dim(\ker(A))<\frac{\max\{k,m\}}{\min\{k,m\}}$ (theorems \ref{theoremlowkernel} and \ref{theoremlowkernel2}). 

We shall adopt the following notation.

\textbf{Notation:} Let $M_{k\times m}$ denote the set of complex matrices with $k$ rows and $m$ columns and $M_k=M_{k\times k}$. Denote by $\|A\|_2$ the spectral norm of the square matrix $A\in M_k$, $\Im(A)$ its image (range) and $\ker(A)$ its kernel. 
Denote by $P_k$ the set of positive semidefinite Hermitian matrices of order $k$ and $VM_kV$ the set $\{VXV, X\in M_k\}$, where $V\in M_k$ is an orthogonal projection. Let $A^{\perp}$ be the orthogonal projection onto $\ker(A)$, where $A\in M_k$.  Let $(x_i)_{i=1}^k$ denote a column vector. If $x_i>0$, for every $i$, then we shall say that $(x_i)_{i=1}^k$ is  a positive vector.  Let   $A\odot B$  denote the Hadamard product (the coordinatewise product) and $A\otimes B$ the Kronecker product of the matrices $A,B$. We shall denote by $1_{m\times k}$ the matrix in $M_{m\times k}$ with all entries equal to $1$.
 Define  $\sigma(A)=\prod_{i=1}^ka_{i\sigma(i)}$, where $A=(a_{ij})$ is a matrix of order $k$ and $\sigma$ a permutation of $S_k$.  If $\alpha\subset \{1,\ldots,k\}$ and $\beta\subset \{1,\ldots,m\}$  are non empty subsets then $A[\alpha|\beta]$  denotes the submatrix of $A\in M_{k\times m}$ using rows $\alpha$ and columns $\beta$, $A(\alpha|\beta)$ denotes the submatrix of $A$ using rows and columns complementaries to $\alpha,\beta$, respectively, and $|\alpha|$ shall denote the cardinality of $\alpha$.  Let $A=\bigoplus_{i=1}^s A[\alpha_i|\beta_i]\in M_{k\times m}$  be such that $a_{ij}=0$, if $(i,j)\notin \cup_{i=1}^s\alpha_i\times \beta_i$ and $\{1,\ldots,k\}=\bigcup_{i=1}^s\alpha_i$, $\{1,\ldots,m\}=\bigcup_{i=1}^s\beta_i$, where the sets $\alpha_i$, $1\leq i\leq s$, are disjoint and non empty, and the sets $\beta_i$, $1\leq i\leq s$, are disjoint and non empty. This matrix shall be called the direct sum of $A[\alpha_i|\beta_i]$, $1\leq i\leq s$. Let $\langle A,B \rangle=tr(AB^*)$ for $A,B\in M_k$.

\section{A Slight Modification of Sinkhorn-Knopp ideas}

The concepts of support and total support for square matrices play a very important role in Sinkhorn-Knopp theorem. 
In this section we extend these notions to rectangular matrices and we adapt one key lemma (lemma \ref{lemmasinkhorn}), used by them in order to obtain their result, for rectangular matrices. In the next section, we define positive maps with support and total support.

\begin{definition}\label{definitionsupportsquare}  We say that $A=(a_{ij})\in M_k$ has support, if there is a permutation $\sigma\in S_k$ such that $\sigma (A)\neq 0$. We say that $A$ has total support, if for every $a_{i_0j_0}\neq 0$, there is a permutation $\sigma\in S_k$ such that $\sigma(i_0)=j_0$ and $\sigma(A)=\prod_{i=1}^ka_{i\sigma(i)}\neq 0$, or equivalently, the matrix $A(\{i_0\}|\{j_0\})$ has support. 
\end{definition} 


One way to extend the ideas of support and total support to rectangular matrices is the following definition. 

\begin{definition}\label{definitionsupportnonsquare}  We say that $A=(a_{ij})\in M_{k\times m}$ has support $($total support$)$, if $A\otimes 1_{m\times k}\in M_{km}$ has support $($total support$)$.
\end{definition}

This extension  is quite  natural, since $A\in M_{k}$ has support (total support) if and only if $A\otimes 1_{k\times k}\in M_{k^2}$ has support (total support) by item $(3)$ of  lemma \ref{lemmasupport}. In order to prove this  lemma, we shall need the following  result and a very simple corollary. The reader can find its proof in \cite[pg 97]{Marcus}.  

\begin{theorem}\label{thmFrobenius}$($Frobenius-K\"onig Theorem$)$ The matrix $A\in M_k$ does not have support if and only if there is an identically zero submatrix $A[\alpha|\beta]$  such that $|\alpha|+|\beta|> k$.
\end{theorem}

A very simple corollary of this theorem is the following.

\begin{corollary}\label{corollaryeasy}  $A\in M_k$ does not have total support if and only if there is an identically zero submatrix $A[\alpha|\beta]$  such that  or $|\alpha|+|\beta|> k$ or $|\alpha|+|\beta|=k$ and $A(\alpha|\beta)$ is not identically zero.
\end{corollary}

The next lemma extend these two results to rectangular matrices.

\begin{lemma}\label{lemmasupport} Let $A\in M_{k\times m}$. Then
\begin{enumerate}
\item $A$ does not have support if and only if there is an identically zero submatrix $A[\alpha|\beta]$ such that $|\alpha|m+|\beta|k>km$.
\item $A$ does not have total support if and only if there is an identically zero submatrix $A[\alpha|\beta]$ such that or  $|\alpha|m+|\beta|k>km$ or $|\alpha|m+|\beta|k=km$ and $A(\alpha|\beta)$ is not identically zero.
\item If $k= m$ then $A$ has support $($total support$)$ if and only if $A\otimes 1_{k\times k}$ has support $($total support$)$.
\item If $k$ and $m$ are coprime then $A$ has support if and only if $A$ has total support.
\item If $k\neq m$ and the cardinality of $\{(i,j)|\ a_{ij} =0 \}<\min\{k,m\}$ then $A$ has total support.
\item If $k= m$ and the cardinality of $\{(i,j)|\ a_{ij} = 0 \}<k-1$ then $A$ has total support.
\item If $A$ has no column or row identically zero and the cardinality of $\{(i,j)|\ a_{ij} = 0 \}<\frac{\max\{k,m\}}{\min\{k,m\}}$ then $A$ has total support.
\item If $A=\bigoplus_{i=1}^s A[\alpha_i|\beta_i]$, $\frac{|\beta_i|}{|\alpha_i|}=\frac{m}{k}$  and $A[\alpha_i|\beta_i]$ has support $($total support$)$, $1\leq i\leq s$, then $A$ has support $($total support$)$.
\end{enumerate}
\end{lemma}

\begin{proof}
  Let $C=A\otimes 1_{m\times k}$. Notice that any identically zero $C[\alpha'|\beta']$ is a submatrix of some identically zero $A[\alpha|\beta]\otimes 1_{m\times k}$. Since any  $A[\alpha|\beta]\otimes 1_{m\times k}$ is also a submatrix of $C$ then the identically zero matrices $C[\alpha'|\beta']$ with maximum $|\alpha'|+|\beta'|$ are the identically zero matrices $A[\alpha|\beta]\otimes 1_{m\times k}$ with maximum $|\alpha|m+|\beta|k$.\\\\
$(1)$ $C=A\otimes 1_{m\times k}$ has no support if and only if there is an identically zero submatrix $C[\alpha'|\beta']=A[\alpha|\beta]\otimes 1_{m\times k}$ such that $|\alpha'|+|\beta'|=|\alpha|m+|\beta|k>km$, by theorem \ref{thmFrobenius}.\\\\
$(2)$ $C=A\otimes 1_{m\times k}$ does not have total support if and only if there is an identically zero submatrix $C[\alpha'|\beta']=A[\alpha|\beta]\otimes 1_{m\times k}$ such that or $|\alpha'|+|\beta'|=|\alpha|m+|\beta|k>km$ or $|\alpha'|+|\beta'|=|\alpha|m+|\beta|k=km$ and $C(\alpha'|\beta')=A(\alpha|\beta)\otimes 1_{m\times k}$ is not identically zero, by corollary \ref{corollaryeasy}. \\\\
$(3)$ Let $k=m$ in the proofs of items $(1)$, $(2)$ and then use theorem \ref{thmFrobenius} and corollary \ref{corollaryeasy}.\\\\
 $(4)$ If $k$ and $m$ are coprime then there are no positive integers $x,y$ such that $xm+yk=mk$. Now, use $(1)$ and $(2)$. \\\\
$(5),(6)$ Let $C[\alpha'|\beta']$ be  identically zero with maximum $|\alpha'|+|\beta'|$ then $C[\alpha'|\beta']=A[\alpha|\beta]\otimes 1_{m\times k}$. Thus, $|\alpha'|+|\beta'|=|\alpha|m+|\beta|k=\frac{|\alpha|m+|\beta|k}{|\alpha|+|\beta|}(|\alpha|+|\beta|)\leq \frac{|\alpha|m+|\beta|k}{|\alpha|+|\beta|}(|\alpha||\beta|+1).$

Since $A[\alpha|\beta]$ is identically zero then $|\alpha||\beta|$ is smaller than $\min\{k,m
\}$, for $k\neq m$, or smaller than $k-1$, if $k= m$, by hypothesis. 
Therefore, if $k\neq m$ then $\frac{|\alpha|m+|\beta|k}{|\alpha|+|\beta|}(|\alpha||\beta|+1)<\max\{m,k\}\times \min\{k,m\}=km$. If $k=m$ then $\frac{|\alpha|m+|\beta|k}{|\alpha|+|\beta|}(|\alpha||\beta|+1)\leq k(k-1)<k^2$. The results follow by $(2)$.\\\\
$(7)$ Let $A[\alpha|\beta]$ be identically zero.  Since there are no rows or columns identically zero then $|\alpha|\leq k-1$  and $|\beta|\leq m-1$.  

Next, if $|\alpha|m+|\beta|k\geq km$ then  $(k-1)m+|\beta|k\geq km$ and $|\alpha|m+(m-1)k\geq km$. Thus, $|\beta|\geq \frac{m}{k}$ and $|\alpha|\geq\frac{k}{m}$, which is impossible, since the cardinality of $\{(i,j)|\ a_{ij} = 0 \}<\frac{\max\{k,m\}}{\min\{k,m\}}$. Therefore, $|\alpha|m+|\beta|k< km$ and $A$ has total support by item (2).\\\\
$(8)$ Notice that $A\otimes 1_{m\times k}=\bigoplus_{i=1}^s (A[\alpha_i|\beta_i]\otimes 1_{m\times k})$ and $|\alpha_i|m=|\beta_i|k$ for every $i$. So $A\otimes 1_{m\times k}$ is a direct sum of square matrices with support (total support) then $A$ has support (total support). 
\end{proof}

The next lemma is a slight modification of  \cite[Lemma 2]{Sinkhorn}.

\begin{lemma}\label{lemmasinkhorn}  Let $(v_n)_{n\in\mathbb{N}}$ be a  sequence of positive vectors of $\mathbb{R}^k$ and  $(w_n)_{n\in\mathbb{N}}$  a sequence of positive vectors of $\mathbb{R}^m$. 
Let $A=(a_{ij})\in M_{k\times m}$ be a matrix with total support. If for every $a_{ij}\neq 0$ the corresponding entry of $v_nw_n^t($i.e., $(v_nw_n^t)_{i,j}=v_{i,n}w_{j,n})$ converges to a positive limit then  there are two sequences of positive vectors $(v_n')_{n\in\mathbb{N}}$ of $\mathbb{R}^k$ and  $(w_n')_{n\in\mathbb{N}}$ of $\mathbb{R}^m$  converging to positive vectors $v,w$, respectively, such that
$v_nw_n^t=v_n'w_n'^t$ for every $n$ $($i.e., $v_{i,n}w_{j,n}=v_{i,n}'w_{j,n}'$, for every $i,j,n)$.

 
\end{lemma}
\begin{proof} 
Let us assume $k\neq m$, since the square case is  \cite[Lemma 2]{Sinkhorn}. 

By definition \ref{definitionsupportnonsquare}, the square matrix $A\otimes 1_{m\times k}$ has total support.

Notice that whenever an entry of the matrix $A\otimes 1_{m\times k}$ is not zero,  the corresponding entry of the rank 1 matrix $v_nw_n^t\otimes 1_{m\times k}=(v_n\otimes 1_{m\times 1})(w_n^t\otimes 1_{1\times k})$ converges to a positive number.  

Since $A\otimes 1_{m\times k} \in M_{km\times km}$ is square with total support and $v_n\otimes 1_{m\times 1},w_n\otimes 1_{k\times1}\in \mathbb{R}^{km}$ are positive vectors then, by the square case of this lemma, there are two  sequences $(v_n')_{n\in\mathbb{N}}$, $(w_n')_{n\in\mathbb{N}}$ of positive vectors of $\mathbb{R}^{km}$ converging to positive vectors $v,w$, respectively, such that $(v_n\otimes 1_{m\times 1})(w_n^t\otimes 1_{1\times k})=v_n'w_n'^t$ for every $n$. Thus, there are subvectors $v_n''\in\mathbb{R}^k$  of $v_n'$ and  $w_n''\in\mathbb{R}^m$ of $w_n'$ converging to positive vectors $v''\in\mathbb{R}^k$, $w''\in\mathbb{R}^m$, respectively, such that $v_nw_n^t=v_n''w_n''^t$, for every $n$.
\end{proof}

The next lemma shall be used later. Its proof is left to the reader.

\begin{lemma}\label{propsigma} 
 \begin{enumerate}
\item If $A,B$ are matrices of order $k$ then $\sigma(A\odot B)=\sigma(A)\sigma(B)$,  for any $\sigma \in S_k$.
\item If $v=(v_i)_{i=1}^k$,  $w=(w_i)_{i=1}^k$ then $\sigma(vw^t)=\prod_{i=1}^kv_iw_i$,  for any $\sigma\in S_k$. 
\item If $v=(v_i)_{i=1}^k$,  $r=(r_i)_{i=1}^m$ then $\sigma(vr^t\otimes 1_{m\times k})=(\prod_{i=1}^kv_i)^m(\prod_{i=1}^mw_i)^k$,  for any $\sigma\in S_{mk}$.
\end{enumerate}
\end{lemma}

\section{Sinkhorn-Knopp algorithm for positive maps}

In this section, we discuss an adaptation  of Sinkhorn-Knopp algorithm for positive maps $T:M_k\rightarrow M_m$. Similar adaptations were used  in \cite{gurvits2004, filternormalform}. This process is commonly known as scaling.

Here, we define positive maps with support and total support (definition \ref{definitionsupportformaps}) and we show that if $T:M_k\rightarrow M_m$ has support then the limit points of the sequence produced by scaling  are doubly stochastic.

\begin{definition} \label{definitiondoublystochastic}
  Let $T: M_k\rightarrow M_m$ be a positive map. We say that $T$ is doubly stochastic if $T(\frac{Id}{\sqrt{k}})=\frac{Id}{\sqrt{m}}$ and $T^*(\frac{Id}{\sqrt{m}})=\frac{Id}{\sqrt{k}}$, where $T^*$ is the adjoint of $T$ with respect to the trace inner product.
  \end{definition}

 The interested reader can find more information concerning  doubly stochastic  maps  in \cite{Choi, gurvits2003, gurvits2004, Landau}.

\begin{definition} \label{definitionsupportformaps}
We say that $T: M_k\rightarrow M_m$ has support $($total support$)$, if for every  orthonormal bases $\{v_1,\ldots,v_k\}$ of $\mathbb{C}^{k}$ and $\{w_1,\ldots,w_m\}$ of $\mathbb{C}^{m}$,  the matrix $(tr(T(v_i\overline{v}_i^t)w_j\overline{w}_j^t))\in M_{k\times m}$  has support $($total support$)$. More generally, we  say that  $T: VM_kV\rightarrow WM_mW$ has support $($total support$)$, where $V\in M_k$ and $W\in M_m$ are orthogonal projections, if for every orthonormal bases $\{v_1,\ldots,v_{k'}\}$  of $\Im(V)$ and $\{w_1,\ldots,w_{m'}\}$  of $\Im(W)$, the matrix $(tr(T(v_i\overline{v}_i^t)w_j\overline{w}_j^t))\in M_{k'\times m'}$  has support $($total support$)$.
\end{definition}

The next lemma provides another description of positive maps $T: M_k\rightarrow M_m$ with support and total support. A similar description is valid for positive maps $T: VM_kV\rightarrow WM_mW$ (See remark \ref{remark2}).

\begin{lemma}\label{lemmaequivalentsupport} Let $T: M_k\rightarrow M_m$  be a positive map. Then
\begin{enumerate}
\item $T: M_k\rightarrow M_m$ has support if and only if for every $A\in P_k$, $\text{rank}(A)m\leq \text{rank}(T(A))k$.
\item  $T: M_k\rightarrow M_m$ has total support if and only if for every $A\in P_k$, $\text{rank}(A)m<\text{rank}(T(A))k$ or $\text{rank}(A)m= \text{rank}(T(A))k$ and  $\Im(T(A))=\Im(T(A^{\perp})^{\perp})$. 
\end{enumerate}
\end{lemma}

\begin{proof}$(1)$ Notice that if $A\in P_k$ and $U\in P_k$ is the orthogonal projection onto $\Im(A)$ then  $\Im(T(A))=\Im(T(U))$, since $T: M_k\rightarrow M_m$  is a positive map.  Thus, we just need to prove the inequality for orthogonal projections.

Now, $T$ has no support if and only if there are orthonormal bases $\{v_1,\ldots,v_k\}\subset \mathbb{C}^k$ and $\{w_1,\ldots,w_m\}\subset\mathbb{C}^m$ such that the matrix $A=(tr(T(v_i\overline{v_i}^t)w_j\overline{w}_j^t))\in M_{k\times m}$ has no support. By lemma \ref{lemmasupport}, this is equivalent to the existence of an identically zero submatrix $A[\alpha|\beta]$ such that $|\alpha|m+|\beta|k>km$. Define $V=\sum_{i\in\alpha}v_i\overline{v_i}^t$ and $W=\sum_{j\in\beta}w_j\overline{w}_j^t$. Notice that $A[\alpha|\beta]$ is identically zero if and only if $tr(T(V)W)=0$.

Next, the existence of orthogonal projections $V\in M_k, W\in M_m$ such that $tr(T(V)W)=0$ and $\text{rank}(V)m+\text{rank}(W)k>km$ is equivalent to $\text{rank}(V)m+\dim(\ker(T(V)))k>km$, since $\Im(W)\subset\ker(T(V))$. Finally, since $\text{rank}(V)m>(m-\dim(\ker(T(V))))k=\text{rank}(T(V))k$ then $T$ has no support if and only if there is an orthogonal projection $V\in M_k$ such that $\text{rank}(V)m>\text{rank}(T(V))k$.\\\\
$(2)$ $(\Rightarrow)$ Let us assume that $T$ has total support. 
Therefore $T$ has support and,  by item $(1)$,
\begin{center}
$m(\text{rank}(A))\leq k(\text{rank}(T(A))),$  for every $A\in P_k$. 
\end{center}

Now, let $ A\in P_k$ be such that  $m(\text{rank}(A))= k(\text{rank}(T(A)))$.

Notice that  if $\text{rank}(A)\in\{0,k\}$ then  $\Im(T(A))=\Im(T(A^{\perp})^{\perp})$ . 

Next, let us assume that  $0<\text{rank}(A)<k$.  Since $m(\text{rank}(A))= k(\text{rank}(T(A)))$ then $0<\text{rank}(T(A))<m$.

Let $\{v_1,\ldots,v_k\}\subset\mathbb{C}^k$ be an orthonormal basis of eigenvectors of $A$ and $\{w_1,\ldots,w_m\}\subset\mathbb{C}^m$ be an orthonormal basis of eigenvectors of $T(A)$. Thus, $A=\sum_{i=1}^ka_iv_i\overline{v}_i^t$ and $T(A)=\sum_{j=1}^mb_jw_j\overline{w}_j^t$, where $a_i,b_j$ are the non-negative eigenvalues. 

Let $\alpha=\{i\ |\ a_i>0\}$ and $\beta=\{j\ |\ b_j=0\}$. Notice that $\alpha$ and $\beta$ are non empty, since  $0<\text{rank}(A)<k$ and $0<\text{rank}(T(A))<m$.

Consider the matrix $C=(tr(T(v_i\overline{v}_i^t)w_j\overline{w}_j^t))\in M_{k\times m}$. Notice that $C[\alpha|\beta]$ is identically zero.  Since $C$ has total support, by hypothesis, and  $$|\alpha|m+|\beta|k=\text{rank}(A)m+(m-\text{rank}(T(A)))k=mk$$ 
 then $C(\alpha|\beta)$ must be identically zero, by item $(2)$ of lemma \ref{lemmasupport}. Hence, $tr(T(A^{\perp})T(A))=0$. So $\Im(T(A))\subset\ker(T(A^{\perp}))=\Im(T(A^{\perp})^{\perp})$.

Now, since  $ \text{rank}(T(A^{\perp}))k\geq \text{rank}(A^{\perp})m$ then
$$\text{rank}(T(A^{\perp})^{\perp})k=(m-\text{rank}(T(A^{\perp})))k\leq mk-\text{rank}(A^{\perp})m=\text{rank}(A)m\leq \text{rank}(T(A))k.$$  

Therefore, $\Im(T(A))=\Im(T(A^{\perp})^{\perp})$.\\\\
$(\Leftarrow)$  Let $\{r_1,\ldots,r_k\}\subset\mathbb{C}^k$  and $\{s_1,\ldots,s_m\}\subset\mathbb{C}^m$ be any orthonormal bases. Consider the matrix $B=(tr(T(r_i\overline{r}_i^t)s_j\overline{s}_j^t))\in M_{k\times m}$.

If $B[\alpha'|\beta']$ is identically zero then $tr(T(R)S)=0$, where $R=\sum_{i\in\alpha'}r_i\overline{r}_i^t$ and $S=\sum_{j\in\beta'}s_j\overline{s}_j^t$. 

Thus, $\Im(T(R))\subset\ker(S)$ which implies $$\text{rank}(T(R))k\leq (m-\text{rank}(S))k=mk-|\beta'|k.$$

 By assumption, $\text{rank}(R)m\leq \text{rank}(T(R))k$ then  $$|\alpha'|m=\text{rank}(R)m\leq \text{rank}(T(R))k \leq mk-|\beta'|k.$$  
 
 Hence, $|\alpha'|m+|\beta'|k\leq km$.

Now, assume $|\alpha'|m= km-|\beta'|k$. Since $$|\alpha'|m=\text{rank}(R)m\leq \text{rank}(T(R))k\leq (m-\text{rank}(S))k=mk-|\beta'|k$$  then $\text{rank}(R)m=\text{rank}(T(R))k=(m-\text{rank}(S))k$.
 Therefore, by hypothesis, $\Im
(T(R))=\Im(T(R^{\perp})^{\perp})$.

Next, since $\Im(T(R))\subset \ker(S)=\Im(S^{\perp})$ and $$\text{rank}(T(R))=m-\text{rank}(S)=\text{rank}(S^{\perp})$$ then  $\Im(S^{\perp})=\Im(T(R))$. Thus, $\Im(S^{\perp})=\Im(T(R))=\Im(T(R^{\perp})^{\perp})$. Therefore, $tr(T(R^{\perp})S^{\perp})=0$ which is equivalent to $B(\alpha'|\beta')$ being identically zero. So, by item $(2)$ of lemma \ref{lemmasupport}, $B$ has total support. Thus, $T$ has total support.
\end{proof}

\begin{remark}\label{remark2} Notice that if $T:VM_kV\rightarrow WM_mW$ is a positive map such that $\Im(T(V))=\Im(W)$ and $\text{rank}(A)\text{rank}(W)\leq \text{rank}(T(A))\text{rank}(V)$, for every $A\in VM_kV\cap P_k$, then $T$ has support. Moreover, if $\text{rank}(A)\text{rank}(W)<\text{rank}(T(A))\text{rank}(V)$, for every $A\in VM_kV\cap P_k$ with $0<\text{rank}(A)<\text{rank}(V),$ then $T$ has total support.
\end{remark}

\begin{examples}\label{examples1} Let $T: M_k\rightarrow M_m$  be a positive map, $S:M_k\rightarrow M_m$ a doubly stochastic map, $\{v_1,\ldots,v_k\}$ any orthonormal basis of $\mathbb{C}^{k}$ and $\{w_1,\ldots,w_m\}$ any orthonormal basis of $\mathbb{C}^{m}$ .
\begin{enumerate}
\item $S$ has total support, since  $(tr(S(v_i\overline{v}_i^t)w_j\overline{w}_j^t))_{k\times m}\otimes 1_{m\times k}$ is a doubly stochastic matrix  scaled by $\sqrt{km}$ and every doubly stochastic matrix has total support $($See \cite{Sinkhorn}$)$.
\item If $k$ and $m$ are coprime then $T: M_k\rightarrow M_m$ has total support iff $T$ has support, by item $(4)$ of lemma \ref{lemmasupport}.
\end{enumerate}
\end{examples}

Let us describe an adaptation of Sinkhorn and Knopp algorithm for positive maps. 

\begin{algorithm}\label{algorithm}$($The Scaling Algorithm$)$
Let $T: M_k\rightarrow M_m$ be a positive map such that $T(Id)$ and $T^*(Id)$ are positive definite Hermitian matrices. 
Define 
\begin{itemize}
\item[] $X_0=Id\in M_k$, $Y_0=(\frac{Id}{\sqrt{m}})^{\frac{1}{2}}T(X_0\frac{Id}{\sqrt{k}}X_0^*)^{-\frac{1}{2}}$,
\item[] $A_0=X_0^* T^*(Y_0^*\frac{Id}{\sqrt{m}}Y_0)X_0$,
\item[] $X_1=X_0A_0^{-\frac{1}{2}}(\frac{Id}{\sqrt{k}})^{\frac{1}{2}}$
\item[] $B_0=Y_0T(X_1\frac{Id}{\sqrt{k}} X_1^*)Y_0^*$
\end{itemize}

 Notice that $A_0$ is a positive definite Hermitian matrix, since $Y_0^*Y_0$ and $T^*(Id)$ are positive definite Hermitian matrices. Analogously, $B_0$ is a positive definite Hermitian matrix. Notice also that $X_0,Y_0,X_1$ are invertible matrices.

Assume that  $X_n,Y_n,A_n,X_{n+1},B_n$ are defined such that $A_n$, $B_n$ are positive definite Hermitian matrices and 
$X_n,Y_n,X_{n+1}$ are invertible matrices. Define 

\begin{itemize}
\item[]  $Y_{n+1}=(\frac{Id}{\sqrt{m}})^{\frac{1}{2}}B_n^{-\frac{1}{2}}Y_n$,
\item[] $A_{n+1}=X_{n+1}^*T^*(Y_{n+1}^*\frac{Id}{\sqrt{m}}Y_{n+1})X_{n+1}$, 
\item[]  $X_{n+2}=X_{n+1}A_{n+1}^{-\frac{1}{2}}(\frac{Id}{\sqrt{k}})^{\frac{1}{2}}$,
\item[]  $B_{n+1}=Y_{n+1}T(X_{n+2}\frac{Id}{\sqrt{k}}X_{n+2}^*)Y_{n+1}^*$.
\end{itemize} 
 
  Notice that $A_{n+1}$ is a positive definite Hermitian matrix, since $Y_{n+1}^*Y_{n+1}$ and $T^*(Id)$ are positive definite Hermitian matrices. Analogously, $B_{n+1}$ is a positive definite Hermitian matrix. Notice also that $X_{n+1},Y_{n+1},X_{n+2}$ are invertible matrices.\\

\end{algorithm}

\begin{lemma}\label{propertiesalgorithm} Let  $T: M_k\rightarrow M_m$ be a positive map such that $T(Id)$, $T^*(Id)$ are positive definite Hermitian matrices. Let $X_n,A_n,Y_n,B_n$ be as defined in algorithm \ref{algorithm}. Then, 

\begin{enumerate}
\item $Y_nT(X_n\frac{Id}{\sqrt{k}} X_n^*)Y_n^*=\frac{Id}{\sqrt{m}}$, $X_{n+1}^*T^*(Y_n^*\frac{Id}{\sqrt{m}} Y_n)X_{n+1}=\frac{Id}{\sqrt{k}} $, for every $n> 0$

\item $tr(A_n)=\sqrt{k}$, $tr(B_n)=\sqrt{m}$,  for every $n > 0$
\item $0<\det(X_n\otimes Y_n)\leq\det(X_{n+1}\otimes Y_{n+1})$
\item If $(\det(X_n\otimes Y_n))_{n\in\mathbb{N}}$ is bounded then every limit point of $(Y_nT(X_n(\cdot)X_n^*)Y_n^*)_{n\in\mathbb{N}}$ is a doubly  stochastic map.
\end{enumerate}
\end{lemma}
\begin{proof}
$(1)$ Notice that 

\begin{itemize}
\item[] $\begin{array}{c}
Y_{n+1}T(X_{n+1}\frac{Id}{\sqrt{k}}X_{n+1}^*)Y_{n+1}^*=
(\frac{Id}{\sqrt{m}})^{\frac{1}{2}}B_n^{-\frac{1}{2}}Y_nT(X_{n+1}\frac{Id}{\sqrt{k}}X_{n+1}^*)Y_n^*B_n^{-\frac{1}{2}}(\frac{Id}{\sqrt{m}})^{\frac{1}{2}}=\\ (\frac{Id}{\sqrt{m}})^{\frac{1}{2}}B_n^{-\frac{1}{2}}B_nB_n^{-\frac{1}{2}}(\frac{Id}{\sqrt{m}})^{\frac{1}{2}}=\frac{Id}{\sqrt{m}}\ ,
\end{array}\\\\
$
\item[] $
\begin{array}{c}
X_{n+1}^*T^*(Y_{n}^*\frac{Id}{\sqrt{m}}Y_{n})X_{n+1}= (\frac{Id}{\sqrt{k}})^{\frac{1}{2}}A_{n}^{-\frac{1}{2}}X_{n}^*T^*(Y_{n}^*\frac{Id}{\sqrt{m}}Y_{n})X_{n}A_{n}^{-\frac{1}{2}}(\frac{Id}{\sqrt{k}})^{\frac{1}{2}}=\\ (\frac{Id}{\sqrt{k}})^{\frac{1}{2}}A_{n}^{-\frac{1}{2}}A_{n}A_{n}^{-\frac{1}{2}}(\frac{Id}{\sqrt{k}})^{\frac{1}{2}}=\frac{Id}{\sqrt{k}}\ .\\\\
\end{array}
$ 
\end{itemize}
 $(2)$ Notice that 
\begin{itemize}
\item[] $\begin{array}{c}
tr(A_{n})=\sqrt{k}\langle X_{n}^*T^*(Y_{n}^*\frac{Id}{\sqrt{m}}Y_{n})X_{n},\frac
 {Id}{\sqrt{k}}\rangle =\sqrt{k}\langle \frac{Id}{\sqrt{m}}, Y_{n}T(X_{n}\frac{Id}{\sqrt{k}}X_{n}^*)Y_{n}^*\rangle=\\ \sqrt{k}\langle \frac{Id}{\sqrt{m}},\frac{Id}{\sqrt{m}}\rangle=\sqrt{k}\ ,\end{array}\\\\
$
\item[] $
\begin{array}{c}
tr(B_{n})=\sqrt{m}\langle Y_{n}T(X_{n+1}\frac{Id}{\sqrt{k}}X_{n+1}^*)Y_{n} , \frac{Id}{\sqrt{m}}\rangle
=\sqrt{m} \langle \frac{Id}{\sqrt{k}}, X_{n+1}^*T^*(Y_{n}^*\frac{Id}{\sqrt{m}}Y_{n})X_{n+1}\rangle=\\ \sqrt{m} \langle\frac{Id}{\sqrt{k}},\frac{Id}{\sqrt{k}}\rangle=\sqrt{m}\ .\\\\
\end{array}
$ 
\end{itemize}
$(3)$ Since $(\sqrt{k}A_n)\otimes (\sqrt{m}B_n)$ is a positive definite Hermitian matrix then $$\det((\sqrt{k}A_n)\otimes (\sqrt{m}B_n))\leq \left(\frac{tr((\sqrt{k}A_n)\otimes (\sqrt{m}B_n))}{km}\right)^{km}= 1.$$  

Thus, 
$$\begin{array}{r}
\dfrac{\det(X_n\otimes Y_n)}{\det(X_{n+1}\otimes Y_{n+1})}=\det(X_{n+1}^{-1}X_n\otimes Y_n Y_{n+1}^{-1})=\det((\sqrt{k}A_n)^{\frac{1}{2}}\otimes (\sqrt{m}B_n)^{\frac{1}{2}})=\\ \det((\sqrt{k}A_n)\otimes (\sqrt{m}B_n))^{\frac{1}{2}}\leq 1.\\\\
\end{array}
$$

Therefore, $0<\det(X_1\otimes Y_1)\leq \det(X_n\otimes Y_n)\leq\det(X_{n+1}\otimes Y_{n+1})$.
\\\\
$(4)$ Since $(\det(X_n\otimes Y_n))_{n\in\mathbb{N}}$ is bounded then, by item $(3)$,  ${\displaystyle \lim_{n\rightarrow \infty}}\det(X_n\otimes Y_n)=L>0$. Thus,

 $$1={\displaystyle\lim_{n\rightarrow\infty}}\left(\frac{\det(X_n\otimes Y_n)}{\det(X_{n+1}\otimes Y_{n+1})}\right)^2={\displaystyle\lim_{n\rightarrow\infty}}\det(X_{n+1}^{-1}X_n\otimes Y_n Y_{n+1}^{-1})^2={\displaystyle\lim_{n\rightarrow\infty}}\det((\sqrt{k}A_n)\otimes (\sqrt{m}B_n)).$$

Let $C$ be a limit point of the sequence $((\sqrt{k}A_n)\otimes (\sqrt{m}B_n))_{n\in\mathbb{N}}$  (there are limit points since $tr((\sqrt{k}A_n)\otimes (\sqrt{m}B_n))=km$, by item 2, and $A_n,B_n$ are positive definite). Hence, 
 $\det(C)=1$ and $tr(C)=km$. Since $C$ is also positive semidefinite then $C=Id\otimes Id$.
 
  Hence, ${\displaystyle \lim_{n\rightarrow \infty}}(\sqrt{k}A_n)\otimes (\sqrt{m}B_n)=Id\otimes Id$ and $${\displaystyle \lim_{n\rightarrow \infty}}(\sqrt{k}A_n)tr(\sqrt{m}B_n)={\displaystyle \lim_{n\rightarrow \infty}}(\sqrt{k}A_n)m=mId.$$

 So, ${\displaystyle \lim_{n\rightarrow \infty}}A_n=\frac{Id}{\sqrt{k}}$.

 Since the operator norm of a positive map induced by the spectral norm is attained at the identity (\cite[corollary 2.3.8]{Bhatia1}) then $\|Y_nT(X_n(\cdot)X_n^*)Y_n^*\|=\| \sqrt{k}Y_nT(X_n\frac{Id}{\sqrt{k}}X_n^*)Y_n^*\|_2=\|\frac{\sqrt{k} Id}{\sqrt{m}}\|_2=\frac{\sqrt{k}}{\sqrt{m}}$. Thus, there are limit points of the sequence of positive maps $(Y_nT(X_n(\cdot)X_n^*)Y_n^*)_{n\in\mathbb{N}}$. 
 
 Since $Y_nT(X_n\frac{Id}{\sqrt{k}}X_n^*)Y_n^*=\frac{Id}{\sqrt{m}}$ and $X_{n}^*T^*(Y_{n}^*\frac{Id}{\sqrt{m}}Y_{n})X_{n}=A_{n}\stackrel{n\rightarrow \infty}{\longrightarrow}\frac{Id}{\sqrt{k}}$ then these limit points are doubly stochastic. 
\end{proof}
\hspace{0.5 cm}
\begin{lemma}\label{lemmacommute} Let  $T: M_k\rightarrow M_m$ be a positive map such that $T(Id)$, $T^*(Id)$ are positive definite Hermitian matrices. Let $X_{n}\in M_k,Y_{n}\in M_m$ be the matrices defined in algorithm \ref{algorithm}. Assume that there are orthogonal projections $V_i\in M_k$, $W_i\in M_m$, $1\leq i\leq s$, such that
\begin{enumerate}
\item $V_iV_j=0$, $W_iW_j=0$, for $i\neq j$, 
\item $Id_{k\times k}=\sum_{i=1}^sV_i$, $Id_{m\times m}=\sum_{i=1}^s W_i$,
\item $T(V_iM_kV_i)\subset W_iM_mW_i$, for $1\leq i\leq s$.
\end{enumerate}  Then $X_nV_i=V_iX_n$ and $Y_nW_i=W_iY_n$, for every $i,n$.
\end{lemma}
\begin{proof} Let $[C,D]=CD-DC$.
Notice that $\Im(T(V_i))=\Im(W_i)$ for every $i$, since $T(Id)$ is positive definite and  $T(V_iM_kV_i)\subset W_iM_mW_i$ for every $i$.

Next, $\langle T^*(W_j),\sum_{i\neq j}V_i\rangle=\langle W_j,\sum_{i\neq j} T(V_i)\rangle=0$. Thus,  $\Im(T^*(W_j))=\Im(V_j)$ and  $T^*(W_jM_mW_j)\subset V_jM_kV_j$, for every $j$.

Since $X_0=Id$ then $[V_i,X_0]=0$ for every $i$. Since $T(\frac{Id}{\sqrt{k}})=\sum_{i=1}^s\frac{1}{\sqrt{k}}T(V_i)$ and $T(V_i)\in W_iM_mW_i$  then $[T(\frac{Id}{\sqrt{k}}),W_i]=0$, for every $i$. 

Now, $Y_0=(\frac{Id}{\sqrt{m}})^{\frac{1}{2}}T(\frac{Id}{\sqrt{k}})^{-\frac{1}{2}}$ is a polynomial of $T(\frac{Id}{\sqrt{k}})$, since $T(\frac{Id}{\sqrt{k}})$ is a positive definite Hermitian matrix. Therefore,  $[Y_0,W_i]=0$, for every $i$.

By induction on $n$, assume that $[X_n,V_i]=0$ and $[Y_n,W_i]=0$, for every $i$.
Thus, $[X_n^*,V_i]=0$ and $[Y_n^*,W_i]=0$, for every $i$. 

Therefore, $Y_n^*Y_n\in \bigoplus_{i=1}^s W_iM_mW_i$ and   $$\displaystyle T^*(Y_{n}^*\frac{Id}{\sqrt{m}}Y_{n})\in \bigoplus_{i=1}^s V_iM_kV_i .$$

Next, since $[X_n,V_i]=[X_n^*,V_i]=0$ and  $T^*(Y_{n}^*\frac{Id}{\sqrt{m}}Y_{n})\in \bigoplus_{i=1}^s V_iM_kV_i $ then $[A_n,V_i]=0$, for every $i$, where $A_n=X_{n}^*T^*(Y_{n}^*\frac{Id}{\sqrt{m}}Y_{n})X_{n}$.

Notice that $A_n^{-\frac{1}{2}}$ is a polynomial of $A_n$,  since $A_n$ is a positive definite Hermitian matrix. Hence, $[X_{n+1},V_i]=[X_{n+1}^*,V_i]=0$, for every $i$, where $X_{n+1}=X_{n}A_{n}^{-\frac{1}{2}}(\frac{Id}{\sqrt{k}})^{\frac{1}{2}}$.

Therefore, $X_{n+1}X_{n+1}^*\in \bigoplus_{i=1}^s V_iM_kV_i$ and $$T(X_{n+1}\frac{Id}{\sqrt{k}} X_{n+1}^*) \in \bigoplus_{i=1}^s W_iM_mW_i.$$

Next, since $[Y_n,W_i]=[Y_n^*,W_i]=0$, for every $i$,  and $T(X_{n+1}X_{n+1}^*) \in \bigoplus_{i=1}^s W_iM_mW_i$ then $[B_n,W_i]=0$, for every $i$, where $B_{n}=Y_{n}T(X_{n+1}\frac{Id}{\sqrt{k}}X_{n+1}^*)Y_{n}^*$.

 Finally,  $B_n$ is a positive definite Hermitian matrix, hence $B_n^{-\frac{1}{2}}$ is a polynomial of $B_n$. Therefore, $[Y_{n+1},W_i]=0$, for every $i$, where $Y_{n+1}=(\frac{Id}{\sqrt{m}})^{\frac{1}{2}}B_n^{-\frac{1}{2}}Y_n$. The induction is complete.
\end{proof}

\vspace{0.1 cm}

\section{Main Results}

\vspace{0.3 cm}

In this section, we show that if  $T:M_k\rightarrow M_m$ has total support then there are invertible matrices $X'\in M_k$, $Y'\in M_m$  such that $Y'T(X'XX'^*)Y'^*$ is doubly stochastic (lemma \ref{keylemma}).  Differently from the classical Sinkhorn-Knopp theorem, the condition of total support is not necessary for the equivalence  of a positive map with a doubly stochastic one (See remark \ref{remark3}).

As a  consequence of  lemma \ref{keylemma}, we obtain a necessary and sufficient condition for the equivalence of a positive map with a doubly stochastic map (theorems \ref{maintheorem1} and \ref{maintheorem}).

\vspace{0.3 cm}

\begin{lemma}\label{keylemma}
Let  $T: M_k\rightarrow M_m$ be a positive map such that $T(Id)$, $T^*(Id)$ are positive definite Hermitian matrices. Let $T_1: M_k\rightarrow M_m$ be any limit point of the sequence of positive maps $(Y_{n}T(X_{n}(\cdot)X_{n}^*)Y_{n}^*)_{n\in\mathbb{N}}$, where $X_{n}\in M_k,Y_{n}\in M_m$ are defined in algorithm \ref{algorithm}. Let $V_i\in M_k$, $W_i\in M_m$, $1\leq i\leq s$, be orthogonal projections such that $V_iV_j=0$, $W_iW_j=0$, for $i\neq j$, and $Id=\sum_{i=1}^sV_i$, $Id=\sum_{i=1}^s W_i$ and $T(V_iM_kV_i)\subset W_iM_mW_i$, for $1\leq i\leq s$. Let $\text{rank}(W_i)=w_i$ and $\text{rank}(V_i)=v_i$.

\begin{enumerate}
\item If $T:V_iM_kV_i\rightarrow W_iM_mW_i$ has support and $\frac{w_i}{v_i}=\frac{m}{k}$, for every $i$, then $T_1$ is doubly stochastic.
\item If $T:V_iM_kV_i\rightarrow W_iM_mW_i$ has total support and $\frac{w_i}{v_i}=\frac{m}{k}$, for every $i$,  then there are invertible matrices $X'\in M_k$, $Y'\in M_m$  such that $T_1(X)=Y'T(X'XX'^*)Y'^*$.
\end{enumerate}
\end{lemma}
\begin{proof} 
Let $X_{n}=L_{n}D_{n}M_{n}^*$, $Y_{n}=S_{n}\tilde{D}_{n}R_{n}^*$, where $L_{n},M_{n}\in M_k$, $S_{n}, R_{n}\in M_m$ are unitary matrices, and $D_{n}=\text{diag}(x_{1,n},\ldots, x_{k,n})$, $\tilde{D}_{n}=\text{diag}(y_{1,n},\ldots, y_{m,n})$ are positive diagonal matrices (i.e., SVD decompositions of $X_{n}$ and $Y_{n}$).

By lemma \ref{lemmacommute}, $X_nV_a=V_aX_n$ and $Y_nW_a=W_aY_n$ for $1\leq a\leq s$. Hence,  we can assume without loss of generality that the columns 
 $v_1+\ldots+v_{a-1}+1,\ldots,v_1+\ldots+v_{a}$ of $L_n$ and the columns $w_1+\ldots+w_{a-1}+1,\ldots, w_1+\ldots+w_{a}$ of $R_n$ form orthonormal bases of $\Im(V_a)$ and $\Im(W_a)$, respectively, for $1\leq a\leq s$ $($when $a=1$, let $v_0=w_0=0)$. 
 

Since the set of unitary matrices is compact,  we can pass to a subsequence to ensure the convergence of $L_{n},M_{n},S_{n},R_{n}$ to unitary matrices $L,M,S,R$, respectively. In order to simplify our notation, we shall write \begin{center}
$\displaystyle\lim_{n\rightarrow\infty} L_{n}=L, \displaystyle\lim_{n\rightarrow\infty} M_{n}=M$, $\displaystyle\lim_{n\rightarrow\infty} S_{n}=S$, $\displaystyle\lim_{n\rightarrow\infty} R_{n}=R$, $\displaystyle\lim_{n\rightarrow\infty} Y_nT(X_nXX_n^*)Y_n^*=T_1(X)$. 
\end{center}

Next, since $\displaystyle\lim_{n\rightarrow\infty} L_{n}=L$ and $\displaystyle\lim_{n\rightarrow\infty} R_{n}=R$ then  the columns $v_1+\ldots+v_{a-1}+1,\ldots,v_1+\ldots+v_{a}$ of $L$ form an orthonormal basis of $\Im(V_a)$ and the columns $w_1+\ldots+w_{a-1}+1,\ldots, w_1+\ldots+w_{a}$ of $R$ form an orthonormal basis of $\Im(W_a)$ for $1\leq a\leq s$.\\

Let $l_{i,n},m_{i,n},s_{i,n},r_{i,n},l_{i},m_{i},s_{i},r_{i}$ be the columns $i$ of $L_n,M_n,S_n,R_n,L,M,S,R$,  respectively.\\

Consider the following matrices of order $mk$ with non-negative entries:$$\begin{array}{l}
\vspace{0.2 cm}  C_n=(tr(Y_nT(X_n m_{i,n}\overline{m}_{i,n}^tX_n^*)Y_n^* s_{j,n}\overline{s}_{j,n}^t))_{k\times m}\otimes 1_{m\times k},\\ \vspace{0.2 cm}  B_n=(tr(T(l_{i,n}\overline{l}_{i,n}^t)r_{j,n}\overline{r}_{j,n}^t))_{k\times m}\otimes 1_{m\times k},\\ A_n=(x_{i,n}^2y_{j,n}^2)_{k\times m}\otimes 1_{m\times k}.
\end{array}$$

Notice that  $C_n=A_n\odot B_n$, since $$x_{i,n}^2y_{j,n}^2tr(T(l_{i,n}\overline{l}_{i,n}^t)r_{j,n}\overline{r}_{j,n}^t)= tr(Y_nT(X_n m_{i,n}\overline{m}_{i,n}^tX_n^*)Y_n^* s_{j,n}\overline{s}_{j,n}^t).$$ 

Moreover, ${\displaystyle\lim_{n\rightarrow \infty}} C_n=C$, ${\displaystyle\lim_{n\rightarrow \infty}} B_n=B$, where $C=(tr(T_1(m_{i}\overline{m}_{i}^t) s_{j}\overline{s}_{j}^t))_{k\times m}\otimes 1_{m\times k}$ and $B=(tr(T(l_i\overline{l}_i^t)r_j\overline{r}_j^t))_{k\times m}\otimes 1_{m\times k}$.

Now, since the columns of $L$ and $R$ have the properties described above and  $T(V_aM_kV_a)\subset W_aM_mW_a$ $(1\leq a\leq s)$ and $W_aW_b=0$ $(a\neq b)$ then $$(tr(T(l_i\overline{l}_i^t)r_j\overline{r}_j^t))_{k\times m}=\bigoplus_{a=1}^s (tr(T(l_i\overline{l}_i^t)r_j\overline{r}_j^t))_{v_{a}\times w_{a}}, $$ i.e., $(tr(T(l_i\overline{l}_i^t)r_j\overline{r}_j^t))_{k\times m}$ is a direct sum of the $s$ matrices  
$(tr(T(l_i\overline{l}_i^t)r_j\overline{r}_j^t))_{v_{a}\times w_{a}}$, where $v_1+\ldots+ v_{a-1}+1\leq i\leq v_1+\ldots+ v_{a}$ and $w_1+\ldots+ w_{a-1}+1\leq j\leq w_1+\ldots+ w_{a}$
 (see the definition of the direct sum in the end of the introduction). 

Notice that if $T:V_aM_kV_a\rightarrow W_aM_mW_a$ has support (total support), for every $a$, then  the matrices  $(tr(T(l_i\overline{l}_i^t))r_j\overline{r}_j^t))_{v_{a}\times w_{a}}$, where $v_1+\ldots+ v_{a-1}+1\leq i\leq v_1+\ldots+ v_{a}$ and $w_1+\ldots+ w_{a-1}+1\leq j\leq w_1+\ldots+ w_{a}$ have support (total support), by definition \ref{definitionsupportformaps}. Since 
$\frac{w_a}{v_a}=\frac{m}{k}$, for every $a$, then the matrix $(tr(T(l_i\overline{l}_i^t))r_j\overline{r}_j^t))_{k\times m}$ has support (total support),  by item 8 of lemma \ref{lemmasupport}. Therefore $B$ has support (total support) by definition \ref{definitionsupportnonsquare}.\\\\
\textbf{Proof of }$(1):$ By item 1 of lemma \ref{propertiesalgorithm}, we have $Y_{n}T(X_nX_n^*)Y_n^*=\frac{\sqrt{k}}{\sqrt{m}}Id$. Thus, $$\frac{\sqrt{k}}{\sqrt{m}}=tr((\frac{\sqrt{k}}{\sqrt{m}}Id)s_{j,n}\overline{s}_{j,n}^t)= \sum_{i=1}^k tr(Y_nT(X_n m_{i,n}\overline{m}_{i,n}^tX_n^*)Y_n^* s_{j,n}\overline{s}_{j,n}^t).$$ 

 So every entry of $C_n$ is smaller or equal to $\frac{\sqrt{k}}{\sqrt{m}}$. Hence, for every permutation  $\sigma'\in S_{mk}$, we have  $$\sigma'(C_n)\leq\left(\frac{\sqrt{k}}{\sqrt{m}}\right)^{mk}.$$
 
Thus, by item 1 of lemma \ref{propsigma}, $\sigma'(A_n)\sigma'(B_n)=\sigma'(C_n)\leq\left(\frac{\sqrt{k}}{\sqrt{m}}\right)^{mk}$ for every  $\sigma'\in S_{mk}$.

Next, since $T:V_aM_kV_a\rightarrow W_aM_mW_a$ has support and $\frac{w_a}{v_a}=\frac{m}{k}$, for every $a$, then $B$ has support. Hence, there is $\sigma\in S_{mk}$, such that $\sigma(B)> 0$, by definition \ref{definitionsupportsquare}. So there is $N>0$, such that if $n>N$ then $\sigma(B_n)>\frac{\sigma(B)}{2}$.
Hence, for $n>N$, $$\sigma(A_n)\leq 2\sigma(B)^{-1}\left(\frac{\sqrt{k}}{\sqrt{m}}\right)^{mk}.$$ 

Now, by item $3$ of lemma \ref{propsigma}, $$ \sigma(A_n)=(\prod_{i=1}^kx_{i,n}^2)^m(\prod_{j=1}^my_{j,n}^2)^{k}=\det(X_nX_n^*)^m\det(Y_nY_n^*)^k
=\det(X_nX_n^*\otimes Y_nY_n^*).$$

 Since $\det(X_n\otimes Y_n)>0$, by item $3$ of lemma \ref{propertiesalgorithm},  then  $\det(X_nX_n^*\otimes Y_nY_n^*)=\det(X_n\otimes Y_n)^2$.

Thus, $\det(X_n\otimes Y_n)^2\leq  2\sigma(B)^{-1}\left(\frac{\sqrt{k}}{\sqrt{m}}\right)^{mk}$ for $n>N$. 

Recall that we simplified the notation in the beginning of this lemma. Thus, we have just proved that there is a bounded subsequence of $(\det(X_{n}\otimes Y_{n}))_{n\in \mathbb{N}}$. Since the entire sequence is increasing, by item 3 of lemma \ref{propertiesalgorithm}, then the entire sequence is bounded. So, by item 4 of lemma \ref{propertiesalgorithm}, $T_1:M_k\rightarrow M_m$ is doubly stochastic.\\\\
\textbf{Proof of }$(2):$ We have just seen that $(\det(X_{n}\otimes Y_{n}))_{n\in \mathbb{N}}$ is bounded and increasing, therefore \begin{center}
${\displaystyle\lim_{n\rightarrow\infty}}\det(X_{n}\otimes Y_{n})=L>0$ and ${\displaystyle\lim_{n\rightarrow\infty}}\sigma(A_n)=L^2$.
\end{center}

Let $tr(T(l_i\overline{l}_i^t)r_j\overline{r}_j^t)$ be any non-null entry of $B$. Since $T:V_aM_kV_a\rightarrow W_aM_mW_a$ has total support and $\frac{w_a}{v_a}=\frac{m}{k}$, for every $a$, then $B$ has total support. There is a permutation $\sigma\in S_{mk}$ such that $\sigma(B)> 0$ and $tr(T(l_i\overline{l}_i^t)r_j\overline{r}_j^t)$ is one of the factors of $\sigma(B)$, by definition \ref{definitionsupportsquare}.

Notice that $tr(T_1(m_i\overline{m}_i^t)s_j\overline{s}_j^t)$ is a factor of $\sigma(C)$, since it occupies the same position of that  $tr(T(l_i\overline{l}_i^t)r_j\overline{r}_j^t)$  in $B$. 

Since $0\neq L^2\sigma(B)={\displaystyle\lim_{n\rightarrow \infty}}\sigma(A_n)\sigma (B_n)=\sigma(C)$ and $tr(T_1(m_i\overline{m}_i^t)s_j\overline{s}_j^t)$ is a factor of $\sigma(C)$ then $tr(T_1(m_i\overline{m}_i^t)s_j\overline{s}_j^t)\neq 0$. Therefore, 
\noindent
\begin{center}
${\displaystyle\lim_{n\rightarrow \infty}}x_{i,n}y_{j,n}= {\displaystyle\lim_{n\rightarrow \infty}}tr(Y_nT(X_n m_{i,n}\overline{m}_{i,n}^tX_n^*)Y_n^* s_{j,n}\overline{s}_{j,n}^t)^{\frac{1}{2}}tr(T(l_{i,n}\overline{l}_{i,n}^t)r_{j,n}\overline{r}_{j,n}^t)^{-\frac{1}{2}}=
tr(T_1(m_i\overline{m}_i^t)s_j\overline{s}_j^t)^{\frac{1}{2}}tr(T(l_i\overline{l}_i^t)r_j\overline{r}_j^t)^{-\frac{1}{2}}\neq 0$.
\end{center}

So $(x_{i,n}y_{j,n})_{k\times m}$  is a rank 1 matrix whose entries are positive and converge to  positive limits, whenever the corresponding entries of the matrix $(tr(T(l_i\overline{l}_i^t)r_j\overline{r}_j^t))_{k\times m}$ are not zero. Since $(tr(T(l_i\overline{l}_i^t)r_j\overline{r}_j^t))_{k\times m}$ has total support then there are sequences of positive numbers $(x_{i,n}')_{n\in\mathbb{N}}$ $(y_{j,n}')_{n\in\mathbb{N}}$ converging to positive limits $x_i'>0$, $y_j'>0$ such that $x_{i,n}'y_{j,n}'=x_{i,n}y_{j,n}$ for every $i,j,n$, by lemma \ref{lemmasinkhorn}.

Define $X'_n=L_n(\text{diag}(x_{1,n}',\ldots,x_{k,n}'))M_n^*$ and $Y'_m=S_n(\text{diag}(y_{1,n}',\ldots,y_{m,n}'))R_n^*$. 

Notice that $\displaystyle\lim_{n\rightarrow \infty}X_n'=X'$ and  $\displaystyle\lim_{n\rightarrow \infty}Y_n'=Y'$, where \begin{center}
$X'=L(\text{diag}(x_{1}',\ldots,x_{k}'))M^*$ and $Y'=S(\text{diag}(y_{1}',\ldots,y_{m}'))R^*$.
\end{center} 

Furthermore $X',Y'$ are invertible matrices.

Finally, for every $\{i,p\}\subset\{1,\ldots,k\}$ and every $\{j,q\}\subset \{1,\ldots,m\}$, we have \vspace{0.1 cm}
$$\vspace{0.2 cm}\begin{array}{cc}
\vspace{0.2 cm} tr(Y_nT(X_n m_{i,n}\overline{m}_{p,n}^tX_n^*)Y_n^*s_{j,n}\overline{s}_{q,n}^t)=& \\  
\vspace{0.2 cm}  x_{i,n}x_{p,n}y_{j,n}y_{q,n}tr(T(l_{i,n}\overline{l}_{p,n}^t)r_{j,n}\overline{r}_{q,n}^t)=& \vspace{0.2 cm} x_{i,n}'y'_{j,n}x_{p,n}'y'_{q,n}tr(T(l_{i,n}\overline{l}_{p,n}^t)r_{j,n}\overline{r}_{q,n}^t)= \\
  &tr(Y_n'T(X_n' m_{i,n}\overline{m}_{p,n}^tX_n'^*)Y_n'^*s_{j,n}\overline{s}_{q,n}^t)\ . 
\end{array} $$


Therefore, $Y_nT(X_n X X_n^*)Y_n^*=Y_n'T(X_n' XX_n'^*)Y_n'^*$,
 for every $X\in M_k$, and $$T_1(X)=\displaystyle\lim_{n\rightarrow \infty}Y_nT(X_n X X_n^*)Y_n^*=\displaystyle\lim_{n\rightarrow \infty}Y_n'T(X_n' X X_n'^*)Y_n'^*=Y'T(X' XX'^*)Y'^*.$$
\end{proof}

The next corollary provides a way to determine whether a positive map $T: M_k\rightarrow M_m$ has support or not.  This is a necessary condition for the equivalence of a positive map with doubly stochastic one (see remark \ref{remark3}). Moreover, it is also sufficient if $k$ and $m$ are coprime (see corollary \ref{corollarycoprimetotalsupport}). 

\begin{corollary}\label{corollaryequivalencesupport}  Let  $T: M_k\rightarrow M_m$ be a positive map such that $T(Id)$, $T^*(Id)$ are positive definite Hermitian matrices.  Let $X_{n},A_n\in M_k,Y_{n}\in M_m$ be the matrices defined in algorithm \ref{algorithm}. Then the following conditions are equivalent:
\begin{enumerate}
\item $T$ has support,
\item \vspace{0.2 cm} $(\det(X_n\otimes Y_n))_{n\in\mathbb{N}}$ is a bounded sequence,
\item $\displaystyle\lim_{n\rightarrow \infty} A_n=\frac{Id}{\sqrt{k}}$.
\end{enumerate}
\end{corollary}

\begin{proof}
In the proof of item $1$ of lemma \ref{keylemma}, we saw that $(1)$ implies $(2)$ (Choose $V_1=Id$ and $W_1=Id$). In the proof of item $4$ of lemma \ref{propertiesalgorithm}, we saw that $(2)$ implies $(3)$. 

Now, let us prove that $(3)$ implies $(1)$.
We also saw in the proof of item $4$ of lemma \ref{propertiesalgorithm} that if $\displaystyle\lim_{n\rightarrow \infty} A_n=\frac{Id}{\sqrt{k}}$ then the limit points of the sequence of positive maps $(Y_{n}T(X_{n}(\cdot)X_{n}^*)Y_{n}^*)_{n\in\mathbb{N}}$ are doubly stochastic.

 Thus,  there is a sequence of positive maps $T_i:M_k\rightarrow M_m$ equivalent to $T$ (i.e., $T_i(X)=Y_{n_i}T(X_{n_i}XX_{n_i}^*)Y_{n_i}^*$) converging to a doubly stochastic map $S:M_k\rightarrow M_m$, which has support (See \ref{examples1}). Notice that if any $T_i$ has support then $T$ should also have support, by lemma \ref{lemmaequivalentsupport}.
 
Let us assume by contradiction that every $T_i$ does not have support. 
Recall that if $A\in P_k$ and $U\in P_k$ is the orthogonal projection onto $\Im(A)$ then  $\Im(T_i(A))=\Im(T_i(U))$, since $T_i: M_k\rightarrow M_m$  is a positive map.  So, for each $T_i$, there is an orthogonal projection $U_i\in M_k$ such that $\text{rank}(U_i)m>\text{rank}(T_i(U_i))k$, by lemma \ref{lemmaequivalentsupport}.

Next, there is a subsequence $(U_j)_j$ of $(U_i)_{i}$ such that $\displaystyle\lim_{j}U_j=U$,
 $\text{rank}(U_j)=u$ and $\text{rank}(T_j(U_j))=t$, for every $j$. Thus, $\displaystyle S(U)=\lim_jT_j(U_j)$ and $\displaystyle\text{rank}(U)=tr(U)=\lim_j tr(U_j)=u$.
 
Since the set of matrices with rank smaller or equal to $t$ is closed then $\text{rank}(S(U))\leq t$. 
Therefore, $\text{rank}(U)m=um>tk\geq \text{rank}(S(U))k$. This is a contradiction, since $S$ has support.
Thus, there is $T_i$ with support and $T$ has also support.
\end{proof}

\begin{remark}\label{remarkpolynomialtime} In \cite{garg2}, a polynomial time algorithm was provided to check whether a completely positive map $T: M_k\rightarrow M_m$ has a positive capacity or not. It turns out that $T: M_k\rightarrow M_m$ has a positive capacity if and only if it has support $($check our lemma \ref{lemmaequivalentsupport} and \cite{garg2}$)$. So for completely positive maps there is an easy way to check whether the map has support or not.
\end{remark}

\begin{theorem}\label{maintheorem1}  Let  $T: M_k\rightarrow M_m$ be a positive map such that $T(Id)$, $T^*(Id)$ are positive definite Hermitian matrices. Then $T:M_k\rightarrow M_m$ is equivalent to a doubly stochastic map if and only if there are orthogonal projections $V_i\in M_k$, $W_i\in M_m$, $1\leq i\leq s$, satisfying
\begin{enumerate}
\item $\mathbb{C}^k=\bigoplus_{i=1}^s\Im(V_i)$, $\mathbb{C}^m=\bigoplus_{i=1}^s\Im(W_i)$,
\item $T(V_iM_kV_i)\subset W_iM_mW_i$, 
\item $\text{rank}(X)\text{rank}(W_i)<\text{rank}(T(X))\text{rank}(V_i)$, if $X\in P_k \cap V_iM_kV_i$ and $0<rank(X)<\text{rank}(V_i)$, 
\item $\frac{\text{rank}(W_i)}{\text{rank}(V_i)}=\frac{m}{k}$, for every $i$.
\end{enumerate}
\end{theorem}
\begin{proof}
$(\Rightarrow)$  Let us assume that there are orthogonal projections $V_i\in M_k,W_i\in M_m$, $1\leq i\leq s$,  satisfying these four conditions. Let $X'\in M_k$ and $Y'\in M_m$ be invertible matrices such that $\Im(X'^{-1}V_i)\perp\Im(X'^{-1}V_j)$ and  $\Im(Y'W_i)\perp\Im(Y'W_j)$, $i\neq j$. Let $\widetilde{V_i}, \widetilde{W_i}$ be the orthogonal projections onto $\Im(X'^{-1}V_i),\Im(Y'W_i)$, respectively.  

Define $\widetilde{T}(X)=Y'T(X'XX'^*)Y'^*$.
Notice that $\widetilde{T}(\widetilde{V_i}M_k\widetilde{V_i})\subset\widetilde{W_i}M_m\widetilde{W_i}$. 
Moreover, if $X\in P_k \cap \widetilde{V_i}M_k\widetilde{V_i}$ and $0<rank(X)<\text{rank}(\widetilde{V_i})$ then $\text{rank}(X)\text{rank}(\widetilde{W_i})<\text{rank}(\widetilde{T}(X))\text{rank}(\widetilde{V_i})$. 
By remark \ref{remark2}, $\widetilde{T}:\widetilde{V_i}M_k\widetilde{V_i}\rightarrow\widetilde{W_i}M_k\widetilde{W_i}$ has total support  for every $i$. Now, since $\frac{\text{rank}(\widetilde{W}_i)}{\text{rank}(\widetilde{V}_i)}=\frac{m}{k}$, for every $i$, then $\widetilde{T}$ is equivalent to a doubly stochastic map, by lemma \ref{keylemma}. Thus, $T$ is equivalent to a doubly stochastic map. \\\\
$(\Leftarrow)$  Let  $S: M_k\rightarrow M_m$ be a doubly stochastic map. We saw in \ref{examples1} that $S$ has total support. Thus, if there is $X\in P_k$ such that  \begin{center}
$\text{rank}(X)m=\text{rank}(S(X))k$ and $0<\text{rank}(X)<k$
\end{center} then  $\Im(S(X^{\perp}))=\Im(S(X)^{\perp})$, by lemma \ref{lemmaequivalentsupport}.

 Therefore, $$\text{rank}(X^{\perp})m=\text{rank}( S(X)^{\perp})k=\text{rank}( S(X^{\perp}))k.$$

Let $V$ be the orthogonal projection onto $\Im(X)$ and $W$ the orthogonal projection onto $\Im(S(X))$.

 Since $S$ is a positive map, $\Im(S(V))=\Im(W)$ and $\Im(S(V^{\perp}))=\Im(S(V)^{\perp})=\Im(W^{\perp})$ then\begin{center}
 $S(VM_kV)\subset WM_mW$ and $S(V^{\perp}M_kV^{\perp})\subset W^{\perp}M_mW^{\perp}$. 
\end{center}

Thus, $S^*(WM_mW)\subset VM_kV$ and $S^*(W^{\perp}M_mW^{\perp})\subset V^{\perp}M_kV^{\perp}$.
 
By the definition of doubly stochastic map, we have \begin{center}
$S(\frac{1}{\sqrt{k}}(V+V^{\perp}))=\frac{1}{\sqrt{m}}(W+W^{\perp})$ and $S^*(\frac{1}{\sqrt{m}}(W+W^{\perp}))=\frac{1}{\sqrt{k}}(V+V^{\perp})$.
\end{center} 

Hence, $S(\frac{V}{\sqrt{k}})=\frac{W}{\sqrt{m}}$, $S(\frac{V^{\perp}}{\sqrt{k}})=\frac{W^{\perp}}{\sqrt{m}}$, $S^{*}(\frac{W}{\sqrt{m}})=\frac{V}{\sqrt{k}}$ and
 $S^{*}(\frac{W^{\perp}}{\sqrt{m}})=\frac{V^{\perp}}{\sqrt{k}}$.\\

Since  $\sqrt{\text{rank}(V)}\sqrt{m}=\sqrt{\text{rank}(W)}\sqrt{k}$ and $\sqrt{\text{rank}(V^{\perp})}\sqrt{m}=\sqrt{\text{rank}(W^{\perp})}\sqrt{k}$ then

$$\begin{array}{lr}
S(\frac{V}{\sqrt{\text{rank}(V)}})=\frac{W}{\sqrt{\text{rank}(W)}}, & S(\frac{V^{\perp}}{\sqrt{\text{rank}(V^{\perp})}})=\frac{W^{\perp}}{\sqrt{\text{rank}(W^{\perp})}},\\
S^{*}(\frac{W}{\sqrt{\text{rank}(W)}})=\frac{V}{\sqrt{\text{rank}(V)}}, &
 S^{*}(\frac{W^{\perp}}{\sqrt{\text{rank}(W^{\perp})}})=\frac{V^{\perp}}{\sqrt{\text{rank}(V^{\perp})}}.
\end{array}$$
 \\ 

Therefore, $S:VM_kV\rightarrow WM_mW$ and $S:V^{\perp}M_kV^{\perp}\rightarrow W^{\perp}M_mW^{\perp}$ are doubly stochastic maps. Now, we can use induction on the rank of $V$ and $V^{\perp}$ in order to find the subalgebras satisfying the conditions of this theorem.

Finally, if $T$ is a positive map equivalent to $S$ then we can easily find the required subalgebras satisfying the required conditions.
\end{proof}

\begin{corollary}\label{corollarysnonquare}
Let $\{e_1,\ldots, e_m\}$ be the canonical basis of $\mathbb{C}^m$. Let  $T: M_k\rightarrow M_m$ be a positive map.  Define $\widetilde{T}: M_m\otimes M_k\rightarrow M_m\otimes M_k$ as $$\widetilde{T}(\sum_{i,j=1}^me_ie_j^t\otimes B_{ij})=T(\sum_{i=1}^mB_{ii})\otimes Id_{k\times k},$$ where $B_{ij}\in M_k$. Then $T: M_k\rightarrow M_m$ is equivalent to a doubly stochastic map if and only if $\widetilde{T}: M_m\otimes M_k\rightarrow M_m\otimes M_k$  is equivalent to a doubly stochastic map .
\end{corollary}

\begin{proof}
$(\Rightarrow)$ If there are invertible matrices $X'\in M_k,Y'\in M_m$ such that $Y'T(X'(\cdot)X'^*)Y'^*$ is doubly stochastic then $\frac{1}{\sqrt{km}}(Y'\otimes Id)\widetilde{T}((Id\otimes X')(\cdot)(Id\otimes X')^*)(Y'\otimes Id)^*$ is also doubly stochastic.\\

$(\Leftarrow)$ Let $\widetilde{T}: M_m\otimes M_k\rightarrow M_m\otimes M_k$ be equivalent to a doubly stochastic map. 

There are orthogonal projections $V_1,\ldots,V_s, W_1,\ldots,W_s$ in $M_{mk}\simeq M_m\otimes M_k$ such that 

\begin{enumerate}
\item $\mathbb{C}^{mk}=\bigoplus_{i=1}^s\Im(V_i)=\bigoplus_{i=1}^s\Im(W_i)$,
\item $\widetilde{T}(V_iM_{mk}V_i)\subset W_iM_{mk}W_i$, 
\item $\text{rank}(X)<\text{rank}(\widetilde{T}(X))$, if $X\in P_{mk} \cap V_iM_kV_i$ and $0<rank(X)<\text{rank}(V_i)$, 
\item $\text{rank}(W_i)=\text{rank}(V_i)$, for every $i$.\\
\end{enumerate}

By lemma \ref{lemmaequivalentsupport}, since $\widetilde{T}: M_m\otimes M_k\rightarrow M_m\otimes M_k$ has support  then  $\text{rank}(V_l)\leq\text{rank}(\widetilde{T}(V_l))$.

Now, since  $\Im(\widetilde{T}(V_l))\subset \Im(W_l)$ and $\text{rank}(W_l)=\text{rank}(V_l)$ then $\text{rank}(\widetilde{T}(V_l))=\text{rank}(W_l)$.\\

Let $V_l=\sum_{i,j=1}^m e_ie_j^t\otimes B^l_{ij}$, where $B^l_{ij} \in M_{k}$. 
Notice that $$\widetilde{T}(Id\otimes \frac{1}{m}(\sum_{i=1}^m B^l_{ii}))=T(\sum_{i=1}^m B^l_{ii})\otimes Id=\widetilde{T}(V_l).$$

Next,  $\text{rank}(Id\otimes \frac{1}{m}(\sum_{i=1}^m B^l_{ii}))\leq\text{rank}(V_l)$, otherwise  $$\text{rank}(Id\otimes \frac{1}{m}(\sum_{i=1}^m B^l_{ii}))>\text{rank}(V_l)=\text{rank}(\widetilde{T}(V_l))=\text{rank}(\widetilde{T}(Id\otimes \frac{1}{m}(\sum_{i=1}^m B^l_{ii}))),$$  contradicting lemma \ref{lemmaequivalentsupport}. 

On the other hand, since $V_l\in P_{mk}$ then $\Im(Id\otimes \frac{1}{m}(\sum_{i=1}^m B^l_{ii}) )\supset\Im (V_l)$.
Therefore, $$\Im(Id\otimes \frac{1}{m}(\sum_{i=1}^m B^l_{ii}) )=\Im (V_l).$$
 
Thus, for every $l\in\{1,\ldots,s\}$, there are orthogonal projections $R_l\in M_k$ and $S_l\in M_m$ such that \begin{center}
$V_l=Id\otimes R_l$ and $W_l=S_l\otimes Id$. 
\end{center}

Since $\mathbb{C}^{mk}=\bigoplus_{i=1}^s\Im(V_i)=\bigoplus_{i=1}^s\Im(W_i)$, by item $(1)$ above, then \begin{center}
$\mathbb{C}^k=\bigoplus_{i=1}^s\Im(R_i)$, $\mathbb{C}^m=\bigoplus_{i=1}^s\Im(S_i)$.
\end{center}

The definition of $\widetilde{T}$ and items $(2),(3)$  imply that, for every $l$, $T(R_lM_kR_l)\subset S_lM_mS_l$ and \begin{center}
$\text{rank}(Y)m<\text{rank}(T(Y))k$, if $Y\in P_k \cap R_lM_kR_l$ and $0<rank(Y)<\text{rank}(R_l)$.
\end{center}

Now, notice that $\frac{\text{rank}(S_l)}{\text{rank}(R_l)}=\frac{m}{k}$ (for every $l$), since \begin{center}
$ \text{rank}(R_l)m=\text{rank}(V_l)=\text{rank}(W_l)=\text{rank}(S_l)k.$ 
\end{center}

Therefore,  $\text{rank}(Y)\text{rank}(S_l)<\text{rank}(T(Y))\text{rank}(R_l)$, for every $Y\in P_k \cap R_lM_kR_l$ such that $0<\text{rank}(Y)<\text{rank}(R_l)$.\\

By theorem \ref{maintheorem1}, $T:M_k\rightarrow M_m$ is equivalent to a doubly stochastic map. 
\end{proof}
\begin{remark} Only at the end of the previous proof, our main theorem (\ref{maintheorem1}) is required: The existence of these subalgebras $R_iM_kR_i$ and $S_iM_mS_i$ implies the equivalence of $T:M_k\rightarrow M_m$ with a doubly stochastic map.  From now on, in order to determine whether a positive map is equivalent to a doubly stochastic map or not, we can use only square maps $(\widetilde{T}: M_m\otimes M_k\rightarrow M_m\otimes M_k)$.
\end{remark}

\begin{theorem}\label{maintheorem} A positive map $T:M_k\rightarrow M_m$ is equivalent to a doubly stochastic map if and only if there are invertible matrices  $X'\in M_k$, $Y'\in M_m$  such that $Y'T(X'(\cdot)X'^*)Y'^*$ has total support and $T(Id)$, $T^*(Id)$ are positive definite Hermitian matrices.
\end{theorem}
\begin{proof}  Since every doubly stochastic map has total support then the existence of $X'\in M_k$, $Y'\in M_m$ such that $Y'T(X'(\cdot)X'^*)Y'^*$ has total support is necessary. 

Now, since $T(Id), T^*(Id)$ are positive definite Hermitian matrices and $X',Y'$ are invertible  matrices then $Y'T(X'X'^*)Y'^*, X'^*T^*(Y'^*Y')X'$ are positive definite Hermitian matrices. Next, if $Y'T(X'(\cdot)X'^*)Y'^*$ has total support then there are invertible matrices $X''\in M_k$ and $Y''\in M_m$ such that $Y''Y'T(X'X''(\cdot)X''^*X'^*)Y'^*Y''^*$ is doubly stochastic, by lemma \ref{keylemma} (choose $V_1=Id$ and $W_1=Id$). 
\end{proof}

\vspace{0.3 cm}

\begin{corollary}\label{corollarycoprimetotalsupport} Let  $T: M_k\rightarrow M_m$ be a positive map with $k,m$  coprime. Then $T:M_k\rightarrow M_m$ is equivalent to a doubly stochastic map if and only if $T: M_k\rightarrow M_m$ has support and $T(Id)$, $T^*(Id)$ are positive definite Hermitian matrices.
\end{corollary}
\begin{proof}
By theorem \ref{maintheorem}, $T:M_k\rightarrow M_m$ is equivalent to a doubly stochastic map if and only if there are invertible matrices  $X'\in M_k$, $Y'\in M_m$  such that $Y'T(X'(\cdot)X'^*)Y'^*$ has total support and $T(Id)$, $T^*(Id)$ are positive definite Hermitian matrices. Since $k$ and $m$ are coprime then support is equivalent to total support (see \ref{examples1}). By lemma \ref{lemmaequivalentsupport}, $Y'T(X'(\cdot)X'^*)Y'^*$ has support if and only if $T:M_k\rightarrow M_m$ has support.
\end{proof}

\vspace{0.3 cm}

\begin{remark}\label{remark3}  Since every doubly stochastic map has support by \ref{examples1} then every positive map equivalent to a doubly stochastic map has also support by lemma \ref{lemmaequivalentsupport}. Thus, the condition of support is necessary for the equivalence of a positive map with a doubly stochastic one. However, the condition of total support is not necessary. For example, let $T:M_2\rightarrow M_2$ be $T(X)=RXR$, where $R=\left(\begin{array}{cc}
0 & 1 \\ 
1 & 1
\end{array} \right)$. This map is clearly equivalent to $Id:M_2\rightarrow M_2$, however it does not have total support.  Notice that if $\{e_1,e_2\}$ is the canonical basis of $\mathbb{C}^2$ then the matrix $(tr(T(e_ie_i^t)e_je_j^t))_{2\times 2}$ is equal to $R$, which does not have total support.
\end{remark}

\vspace{0.3 cm}

\section{The Filter Normal Form}

\vspace{0.3 cm}

 In this section,  we discuss the filter normal form for states that are not positive definite and we provide easy sufficient conditions for the existence of this normal form.

Let $A=\sum_{i=1}^nA_i\otimes B_i\in M_k\otimes M_m\simeq M_{km}$ be a positive semidefinite Hermitian matrix. Define the maps $F_A(X)=\sum_{i=1}^nA_itr(B_iX)$ and $G_A(X)=\sum_{i=1}^nB_itr(A_iX)$.  

We may assume without loss of generality that $A_i,B_i$ are Hermitian for every $i$, since $A$ is Hermitian. Notice that if $A_i,B_i$ are Hermitian for every $i$ then $G_A^*=F_A$ with respect to the trace inner product. Since $A$ is positive semidefinite then $F_A$ and $G_A$ are positive maps (Actually, $G_A(X^t)$ is completely positive by Choi theorem \cite{Choi}, since $A=\sum_{i,j=1}^ke_ie_j^t\otimes G_A((e_ie_j^t)^t)$, where $\{e_1,\ldots,e_k\}$ is the canonical basis of $\mathbb{C}^k)$.

\vspace{0.3 cm}

\begin{theorem} \label{thmequivfilter} Let $A=\sum_{i=1}^nA_i\otimes B_i\in M_k\otimes M_m\simeq M_{km}$ be a positive semidefinite Hermitian matrix such that $G_A(Id)\in M_m$ and $F_A(Id)\in M_k$ are positive definite Hermitian matrices. There are invertible matrices $X'\in M_k$, $Y'\in M_m$ such that $(X'\otimes Y')A(X'\otimes Y')^*=\sum_{i=1}^n C_i\otimes D_i$,   $C_1=\frac{Id}{\sqrt{k}}$, $D_1=\frac{Id}{\sqrt{m}}$ and $tr(C_iC_j)=tr(D_iD_j)=0$, for every $i\neq j$, if and only if there are invertible matrices $X''\in M_k$, $Y''\in M_m$ such that $Y''G_A(X''(\cdot)X''^*)Y''^*$ has total support.
\end{theorem}

\begin{proof} The existence of these matrices $X',Y'$  is equivalent to \begin{center}
$G_{(X'\otimes Y')A(X'\otimes Y')^*}(\frac{Id}{\sqrt{k}})=\frac{Id}{\sqrt{m}}$ and $F_{(X'\otimes Y')A(X'\otimes Y')^*}(\frac{Id}{\sqrt{m}})=\frac{Id}{\sqrt{k}}$.
\end{center} 

Since  $(X'\otimes Y')A(X'\otimes Y')^*$ is Hermitian then $G^*_{(X'\otimes Y')A(X'\otimes Y')^*}=F_{(X'\otimes Y')A(X'\otimes Y')^*}$. 

Therefore the existence of $X',Y'$  is equivalent to $G_{(X'\otimes Y')A(X'\otimes Y')^*}(X)=Y'G_A(X'^*XX')Y'^*$ being a doubly stochastic map, which is equivalent to the existence of invertible matrices $X''\in M_k$, $Y''\in M_m$ such that $Y''G_A(X''(\cdot)X''^*)Y''^*$ has total support, by theorem \ref{maintheorem}.
\end{proof}

\vspace{0.3 cm}

\begin{corollary}\label{corollary0} Let $A=\sum_{i=1}^nA_i\otimes B_i\in M_k\otimes M_m\simeq M_{km}$ be a positive semidefinite Hermitian matrix such that $G_A(Id)\in M_m$ and $F_A(Id)\in M_k$ are positive definite Hermitian matrices. Let $k$ and $m$ be coprime. There are invertible matrices $X'\in M_k$, $Y'\in M_m$ such that $(X'\otimes Y')A(X'\otimes Y')^*=\sum_{i=1}^n C_i\otimes D_i$,   $C_1=\frac{Id}{\sqrt{k}}$, $D_1=\frac{Id}{\sqrt{m}}$ and $tr(C_iC_j)=tr(D_iD_j)=0$, for every $i\neq j$, if and only if $G_A: M_k\rightarrow M_m$ has support.
\end{corollary}
\begin{proof} By theorem \ref{thmequivfilter}, there are such $X'$ and $Y'$ if and only if there are invertible matrices $X''\in M_k$, $Y''\in M_m$ such that $Y''G_A(X''(\cdot)X''^*)Y''^*$ has total support. Since $k$ and $m$ are coprime then support and total support are equivalent (see \ref{examples1}). By lemma \ref{lemmaequivalentsupport}, $Y''G_A(X''(\cdot)X''^*)Y''^*$ has support if and only if $G_A:M_k\rightarrow M_m$ has support.
\end{proof}

\vspace{0.3 cm}

\begin{theorem} \label{theoremlowkernel} Let $A=\sum_{i=1}^nA_i\otimes B_i\in M_k\otimes M_m\simeq M_{km}$ be a positive semidefinite Hermitian matrix. If $k\neq m$ and $\dim(\ker(A))<\min\{k,m\}$ or if $k= m$ and $\dim(\ker(A))<k-1$ then there are invertible matrices $X'\in M_k$, $Y'\in M_m$ such that $(X'\otimes Y')A(X'\otimes Y')^*=\sum_{i=1}^n C_i\otimes D_i$,   $C_1=\frac{Id}{\sqrt{k}}$, $D_1=\frac{Id}{\sqrt{m}}$ and $tr(C_iC_j)=tr(D_iD_j)=0$, for every $i\neq j$.
\end{theorem}

\begin{proof} Let  $ w\in \mathbb{C}^k,\ v\in\mathbb{C}^m$ be unit vectors. Notice that $tr(G_A(Id)v\overline{v}^t)=tr(A(Id\otimes v\overline{v}^t))>0$ and $tr(w\overline{w}^tF_A(Id))=tr(A(w\overline{w}^t\otimes Id))>0$, since $tr(A(Id\otimes v\overline{v}^t))$ is bigger or equal to the sum of the $k$ smallest  eigenvalues of $A$ and $tr(A(w\overline{w}^t\otimes Id))$ is bigger or equal to the sum of the $m$ smallest eigenvalues of  $A$. Therefore $G_A(Id)$ and $F_A(Id)$ are positive definite Hermitian matrices.

Next, let $\{v_1,\ldots,v_k\}\subset \mathbb{C}^k$ and $\{w_1,\ldots, w_m\}\subset \mathbb{C}^m$ be any orthonormal bases. Consider the matrix $(tr(G_A(v_i\overline{v}_i^t)w_j\overline{w}_j^t))\in M_{k\times m}$.

Since $tr(A(v_i\overline{v}_i^t\otimes w_j\overline{w}_j^t))=tr(G_A(v_i\overline{v}_i^t)w_j\overline{w}_j^t)$ then the cardinality of $\{(i,j)|  \ tr(G_A(v_i\overline{v}_i^t)w_j\overline{w}_j^t)= 0\}$ is smaller than $\min\{k,m\}$, if $k\neq m$, or smaller than $k-1$, if $k=m$. By lemma \ref{lemmasupport}, the matrix  $(tr(G_A(v_i\overline{v}_i^t)w_j\overline{w}_j^t))\in M_{k\times m}$ has total support. Therefore $G_A: M_k\rightarrow M_m$ has total support. By theorem \ref{thmequivfilter}, the result follows.
\end{proof}

\vspace{0.3 cm}

\begin{theorem}
\label{theoremlowkernel2} Let $A=\sum_{i=1}^nA_i\otimes B_i\in M_k\otimes M_m\simeq M_{km}$ be a positive semidefinite Hermitian matrix such that $G_A(Id)$ and $F_A(Id)$ are positive definite Hermitian matrices. If  $\dim(\ker(A))<\frac{\max\{k,m\}}{\min\{k,m\}}$ then there are invertible matrices $X'\in M_k$, $Y'\in M_m$ such that $(X'\otimes Y')A(X'\otimes Y')^*=\sum_{i=1}^n C_i\otimes D_i$,   $C_1=\frac{Id}{\sqrt{k}}$, $D_1=\frac{Id}{\sqrt{m}}$ and $tr(C_iC_j)=tr(D_iD_j)=0$, for every $i\neq j$.
\end{theorem}
\begin{proof} Let $\{v_1,\ldots,v_k\}\subset \mathbb{C}^k$ and $\{w_1,\ldots, w_m\}\subset \mathbb{C}^m$ be any orthonormal bases.
Notice that $tr(G_A(v_i\overline{v_i}^t)w_j\overline{w_j}^t)=tr(A(v_i\overline{v_i}^t\otimes w_j\overline{w_j}^t))=tr(v_i\overline{v_i}^tF_A(w_j\overline{w_j}^t))$.

Since $G_A(Id)$ and $F_A(Id)$ are positive definite then the matrix $(tr(G_A(v_i\overline{v_i}^t)w_j\overline{w_j}^t))_{k\times m}$ has no row or column identically zero. 

Next, since $\dim(\ker(A))<\frac{\max\{k,m\}}{\min\{k,m\}}$ then the cardinality of $\{(i,j)|\ tr(G_A(v_i\overline{v_i}^t)w_j\overline{w_j}^t) = 0 \} <\frac{\max\{k,m\}}{\min\{k,m\}}$. Thus, by item 7 of lemma \ref{lemmasupport}, the matrix $(tr(G_A(v_i\overline{v_i}^t)w_j\overline{w_j}^t))_{k\times m}$ has total support. Therefore, $G_A:M_k\rightarrow M_m$ has total support. By lemma \ref{thmequivfilter}, the result follows.
\end{proof}

\section*{Summary and Conclusion}

The search for canonical forms with applications in Quantum Information Theory  is certainly important. In \cite{Chen}, the authors found a canonical form for pure states. In \cite{leinaas, filternormalform}, the filter normal form was obtained for positive definite mixed states.
Here, we provided a necessary and sufficient condition for mixed states in $M_k\otimes M_m$ to be put in the filter normal form. In order to do so, we extended the Sinkhorn and Knopp ideas of support and total support to positive maps and we generalised their result for positive maps $T:M_k\rightarrow M_m$. We showed that a state $A=\sum_{i}^nA_i\otimes B_i\in M_k\otimes M_m$ can be put in the filter normal form if and only if the positive map $G_A: M_k\rightarrow M_m$, defined by $G_A(X)=\sum_{i=1}^nB_itr(A_iX)$, is equivalent to a positive map with total support. When $k$ and $m$ are coprime, the theorem is simpler: A state $A\in M_k\otimes M_m$ with $k$ and $m$ coprime, can be put in the filter normal form if and only if the positive map $G_A: M_k\rightarrow M_m$ has support. We can determine whether a positive map has support or not with a limit.
For the general case, some easy sufficient conditions were provided to guarantee that $A\in M_k\otimes M_m$ can be put in the filter normal form: If $k=m$ and $\dim(\ker(A))<k-1$, or if $k\neq m$ and $\dim(\ker(A))<\min\{k,m\}$, or  if $G_A(Id)$ and $G_A^*(Id)$ are invertible matrices and $\dim(\ker(A))<\frac{\max\{k,m\}}{\min\{k,m\}}$.
It is a surprise to see that the original ideas of Sinkhorn and Knopp  are still very useful in order to generalise their result to positive maps.


\end{document}